\documentclass{article}

\usepackage{arxiv}

\usepackage[utf8]{inputenc} 
\usepackage[T1]{fontenc}    
\usepackage{hyperref}       
\usepackage{url}            
\usepackage{booktabs}       
\usepackage{graphicx}

\usepackage[numbers]{natbib}
\usepackage{amsmath}
\usepackage{amssymb}
\usepackage{tabularx}
\usepackage{fancyhdr}
\usepackage[section]{placeins}
\usepackage{bm}
\usepackage{hyperref}
\usepackage{autonum}
\usepackage{amsthm}

\newtheorem{theorem}{Theorem}[section]
\newtheorem{lemma}[theorem]{\normalfont\scshape Lemma}

\newtheorem{definition}{\normalfont\scshape Definition}[section]

\newcommand{\psd}{\textit{psd}}

\DeclareMathOperator*{\argmin}{arg\,min}

\newcolumntype{Y}{>{\centering\arraybackslash}X}

\title{Feasible Implied Correlation Matrices from Factor Structures}

\author{{\hspace{1mm}Wolfgang Schadner}\thanks{\url{https://sites.google.com/view/wolfgang-schadner/}} \\
	Swiss Institute of Banking and Finance\\
	University of St.Gallen\\
	\texttt{wolfgang.schadner@unisg.ch} 
}

\renewcommand{\shorttitle}{\textit{arXiv} Template}

\hypersetup{
	pdftitle={Feasible Implied Correlation Matrices from Factor Structures},
	pdfsubject={q-fin.MF, q-fin.SF},
	pdfauthor={Wolfgang Schadner},
	pdfkeywords={Implied Correlation , Spectral Projected Gradient , Factor Analysis  , Nearest Correlation Matrix , Inexact Restoration , Positive Semi-Definite},
}

\begin{document}
	\maketitle
	
	\begin{abstract}
		Forward-looking correlations are of interest in different financial applications, including factor-based asset pricing, forecasting stock-price movements or pricing index options. With a focus on non-FX markets, this paper defines necessary conditions for option implied correlation matrices to be mathematically and economically feasible and argues that existing models are typically not capable of guaranteeing so. To overcome this difficulty, the problem is addressed from the underlying factor structure and introduces two approaches to solve it. Under the quantitative approach, the puzzle is reformulated into a nearest correlation matrix problem which can be used either as a stand-alone estimate or to re-establish positive-semi-definiteness of any other model's estimate. From an economic approach, it is discussed how expected correlations between stocks and risk factors (like CAPM, Fama-French) can be translated into a feasible implied correlation matrix. Empirical experiments are carried out on monthly option data of the S\&P 100 and S\&P 500 index (1996-2020).
	\end{abstract}
	
	\keywords{Implied Correlation \and Spectral Projected Gradient \and Factor Analysis  \and Nearest Correlation Matrix \and Inexact Restoration \and Positive Semi-Definite}
	
	\section{Introduction}
	Financial markets are forward-looking by nature, and correlation matrices are at the core of every market. Yet finding a forward-looking solution which is both mathematically and economically feasible remains a challenging task so far (excluding foreign exchange markets, where derivatives are written on currency-pairs; e.g. \cite{Walter2000}). The option implied correlation matrix as such represents the forward-looking dependence between securities (\cite{Skintzi2005, Linders2016}), it is a key requirement for pricing basket/index options (e.g. on the S\&P 500; \cite{Linders2016, Milevsky1998}), but also used for other purposes such as factor-based asset pricing (\cite{Harris2019, Schadner2020b}), forecasting \cite{Skintzi2005,	Driessen2009, Fink2017, Markopoulou2016}, trading strategies (\cite{Buss2012, Kempf2015, Markopoulou2016, Hardle2016}), risk management (\cite{Skintzi2005, Chang2011, Echaust2020}) or understanding market behavior (\cite{Dhaene2012, Just2020}). What makes it favorable in practice is the simultaneous information update with changes in the current market outlook.\footnote{For example, when a market crashes, option prices will react simultaneously, thus implied metrics also change immediately while for backward-looking information it will take hours/days/months to update for the sudden change.} Given the broad interest into implied correlations, there are several reasons why solving the cross-sectional matrix remains a challenging task. The main difficulty is that there are far less options traded as there are unknown correlation-pairs within the matrix, hence the puzzle is under-determined and multiple solutions exist in general. Models in place solve the problem by imposing additional economic assumptions, which, as argued within, cause a rejection of mathematical and/or economical feasibility. Different to that, by generating the implied correlation matrix from factors, this paper proposes two approaches to estimate a fully feasible solution. The presented quantitative approach computes the nearest implied correlation matrix with respect to a pre-specified target, hence can be used as a stand-alone method or to re-establish mathematical feasibility (especially positive semi-definiteness) of any other model's estimate. The economic approach, on the other hand, uses expected correlations between assets to risk factors (like CAPM \cite{Sharpe64}, Fama-French \cite{FF3}) as the determining structure of the implied correlation matrix.\\
	
	More precisely, for a market covering $n$ securities a correlation matrix $C\in\mathbb{R}^{n\times n}$ satisfies mathematical feasibility if (i) it is symmetric, (ii) all its elements lie within the $[\textrm{-}1,1]$ interval, (iii) it has unit diagonal and (iv) it is positive semi-definite (\textit{psd}). Economic feasibility is fulfilled if (v) the matrix is ex-ante free of arbitrage and (vi) the values are realistic in an economic sense. Conditions (i)-(vi) are necessary for $C$ to be a feasible implied correlation matrix. If conditions (i)-(vi) are met, it is also sufficient for $C$ to be a feasible solution of the implied correlation matrix. The definitions of the mathematical conditions are straight forward. From the thousands of publications relating to the factor zoo (e.g., \cite{Cochrane2011, Harvey2015}) it is known that dependencies between securities are of factor structure, hence existence of a factor structure will be used as the proxy of condition (vi). Ex-ante arbitrage-freeness of condition (v) is interpreted in the following sense. As for a market there are typically options traded on individual firms, but also for the market index directly, $C$ will fulfill condition (v) if it properly aggregates the prices of the constituents to match the observable price of the basket option (see e.g., \cite{Skintzi2005, Linders2016}). Hence a minimum requirement for $C$ being an economically feasible implied correlation matrix is
	\begin{align}
		\label{eq:C}
		\text{(v)}: \quad \sigma_m^2 = w' \sigma C \sigma w
	\end{align}
	with $\sigma_m^2\in\mathbb{R}_+$ as the implied variance of the market portfolio, $w\in\mathbb{R}^{n}$ the corresponding weight vector and $\sigma\in \mathbb{R}_+^{n\times n}$ a diagonal matrix of firm implied volatilities. To come up with a solution that matches all six conditions, $C$ will be inherently generated from a $k$-factor structure (\cite{Glasserman2007,Borsdorf2010}), which is derived in its final form as
	\begin{align}
		\label{eq:equilibrium}
		C(X) = J \circ XX' + I
	\end{align}
	having $X\in\mathbb{R}^{n\times k}$ as factor-loadings, $\circ$ the Hadamard product, $I$ the identity matrix and $J:=\mathbf{1}_{n\times n} - I$ with $\mathbf{1}_{n\times n}$ as the unity matrix of dimension as indicated in its subscript. In a factor-based asset pricing model like CAPM (\cite{Sharpe64}) or Fama-French (\cite{FF3, FF5}), $X$ can be interpreted as the correlations between a firm and the factor portfolios, such that introducing a correlation risk premium allows to match the implied variance of the market portfolio. 
	More generally, the question can be also re-formulated as a nearness problem to some target correlation matrix, with the target being freely chosen. In that case, no specific economic assumptions need to be imposed and the feasible solution of is found by a constrained minimization of the distance between $C(X)$ and the target. On an empirical example, this quantitative approach is numerically solved using the spectral projected gradient method with inexact restoration (\cite{Birgin2001, Gomes2009}), which turns out to be computational efficient.\\
	
	
	The reading is set up as follows. Section 2 provides a comparison of existing implied correlation models and argues potential violations of mathematical/economical feasibility. Section 3 discusses the factor model and its implementation for solving the implied correlation matrix. Section 4 illustrates an empirical application of both approaches (quantitative/nearest-implied, economic/risk-factors) of upon data of the S\&P 100 and S\&P 500. Section 5 concludes. Throughout the reading the Hadamard products like $J\circ XX'$ are used, the corresponding product rule should be read as $J \circ (XX')$, not as $(J\circ X)X'$ (i.e., the rule is "matrix-before-Hadamard product").

	\section{Discussion on Existing Models}
	From a literature review I could identify four distinct concepts for solving the implied correlation matrix. Generally, in quantitative finance it is often more convenient to work with implied volatilities directly instead of option prices (cp. \cite{Guo2016}), this also applies for literature on implied correlations. Given non-flat implied volatility surfaces, a common convention is to use options of same strike and maturity (\cite{Skintzi2005, Buss2012, Linders2016}). Within each of the four discussed concepts - and also at the presented framework - the requirement is set that financial options are traded upon a market index $m$ and on each of its constituents, which are weighted according to $w$ to form the market portfolio. From the setting follows, that the equilibrium between the observed portfolio variance $\sigma_m^2$ and its aggregate from constituents (Eq.\ref{eq:equilibrium}) has to hold for every $C$ as a necessary condition for $C$ to be free of arbitrage (cp. \cite{Skintzi2005}). The existing concepts for solving (non-FX) implied correlation matrices are now as follows.
	
	\paragraph{Equi-Correlations} Literature on implied correlations was initialized by \cite{Skintzi2005} in 2005. This model is not only cited in most of the subsequent research in that direction, it is also used as the fundamental concept behind the CBOE Implied Correlation Index. The idea of the method is as follows. Since for a given point in time the weights $w$ and the forward-looking measures $\sigma_m^2$ and $\sigma$ are available (or estimate-able), each off-diagonal element within $C$ can be set to the same scalar $\bar{c}$ to match the equilibrium of Eq.\ref{eq:equilibrium}. The equi-correlation $\bar{c}$ can then easily be computed by
	\begin{align}
		\bar{C}:= \bar{c} J + I: \quad \sigma_m^2 \equiv w' \sigma \bar{C} \sigma w \implies \bar{c} = \frac{\sigma_m^2 - w'\sigma^2 w}{w'\sigma J \sigma w}
	\end{align}
	(cp. \cite{Schadner2020b}). From the proof of \cite{Borsdorf2010}(Sect.2) follows, that $\bar{C}$ will be \textit{psd} if $\textrm{-}1/(n-1) \leq \bar{c} \leq 1$, which is also sufficient for mathematical feasibility. Within an economy, securities are typically on average positively correlated (the CBOE Implied Correlation Index is also thoroughly positive), so this should not be much of an issue empirically. However, it is unrealistic that each correlation pair inside $C$ takes on the same value, hence $\bar{C}$ is seen as economically unfeasible (\cite{Buss2012, Numpacharoen2013}). Also, $\bar{C}$ clearly rejects existence of a (risk-)factor structure (also discussed in \cite{Schadner2020b}), being in contradiction with empirical asset pricing literature. In case the market index also has option traded upon sub-indices (e.g., S\&P 500 sector indices), it is unlikely that the equicorrelation matrix will match them, hence is a further pitfall. Due to those reasons, the equi-correlation approach is more frequently used as an index of average diversification possibilities rather than as correlation matrix (e.g., \cite{Dhaene2012}). 
	
	\paragraph{Local Equicorrelations} The equilibrium (Eq.\ref{eq:equilibrium}) violation from existing sub-index options can be easily resolved by introducing local equicorrelations, where the correlations are averaged out within sub-squares of the matrix. This idea is mentioned in \cite{Buss2016}, who use options on the S\&P 500 (market) and the ten S\&P 500 sector indices (sub-indices). Herein, always the average correlation within a sub-index is computed, before computing the global equicorrelations outside the sub-portfolio. With $w_p$ as the weight of the sub-portfolio $p$ and $\sigma_p^2=w_p'\sigma C \sigma w_p$, local equicorrelations of different $p$'s are always centered alongside the diagonal of $C$. This rises the issue that large parts of $C$ cannot be covered and obviously, with a handful of sub-indices and a very large number of unknown correlation pairs, the resolving correlation matrix will still give a very blur and unrealistic picture.\footnote{For $n=500$, there are 124750 unknown correlation pairs. With 10 sub-indices, there are only 11 distinct estimates} For this definition of local equicorrelations, the \textit{psd} criterion does not necessarily hold. Also, as argued in \cite{Schadner2020b}, the known factor structure of security markets will still be ignored. \cite{Schadner2020b} further demonstrates how the factor-structure can be resolved from calibrating also off-diagonal local equicorrelations. But also this refinement yields in a very blur estimate of the implied correlation matrix. Hence, as the local equicorrelation method only slightly improves precision compared to the global equi-correaltion, it still causes unrealistic estimates of $C$ and is thus economically not feasible.
	
	\paragraph{Adjusted Ex-Post} 
	\cite{Buss2012, Buss2016} combine (local) equicorrelations with backward looking estimates to approximate the forward-looking implied one. To discuss this model, recap that due to the stochastic nature of variances and correlations, option prices are documented to carry risk premia for variance (VRP) and for correlation (CRP; see \cite{Buraschi2010,Buraschi2013, Driessen2013, Faria2018}). Let $\mathbb{P}$ denote investor expectations under the physical- and $\mathbb{Q}$ under the risk-neutral probability measure. Option implied volatilities are known to be $\mathbb{Q}$-measured, and let the physically expected correlation matrix be denoted by $C_\mathbb{P}$, then, the difference between implied market and aggregated individual volatilities can be used to express the volatility quoted ex-ante CRP as $CRP = \sigma_m^2 - w'\sigma C_\mathbb{P} \sigma w$ (cp. \cite{Driessen2009, Buraschi2013}). Building on this idea, let us introduce $A$ to denote a backward-looking estimated correlation matrix, in the example of \cite{Buss2012} it is simply the 1 year historically realized return correlations, but $A$ can also be chosen from a more sophisticated model (e.g., incorporating mean-reversion \cite{Buss2016}). For the model of \cite{Buss2012}, two crucial assumptions are obligatory. First, they introduce that the backward-looking $A$ is equivalent to the forward-looking $\mathbb{P}$ matrix, $A \equiv C_\mathbb{P}$. While investors use backward-looking information to form their believes, there is no fundamental reason for this to hold, as $A$ does not carry information on the market outlook (see e.g., \cite{Skintzi2005}). The second assumption is that the correlation risk premium enters into the matrix in the specific form of 
	\begin{align}
		\label{eq:CQ}
		C_\mathbb{Q} := C_\mathbb{P} - \alpha(\mathbf{1}_{n\times n} - C_\mathbb{P})
	\end{align}
	with $\alpha$ as a scalar calibrating for the correlation risk premium and $C_\mathbb{Q}$ as their estimate of the implied correlation matrix. The two assumptions come at a mathematical convenience. Let $\hat{\alpha}:=\textrm{-}\alpha $, then by rearranging terms the equation can be brought into the form $C_\mathbb{Q} = \hat{\alpha} \mathbf{1}_{n\times n} + (1-\hat{\alpha}) C_\mathbb{P}$, known in literature as 'weighted average correlation matrix' (\cite{Numpacharoen2013b}). Since both $\mathbf{1}_{n\times n}$ and $C_\mathbb{P}$ are \textit{psd}, and the sum of two \textit{psd} matrices is also \textit{psd}, mathematical feasibility of $C_\mathbb{Q}$ holds for $\alpha \in (\textrm{-}1,0]$, which is fulfilled when $CRP>0$. Empirically, however, $\alpha$ is likely to fall outside that range and feasibility of the method does not hold in general.\footnote{On monthly S\&P 100 data from 1996 to 2020 with $A$ as the one year historically realized correlation matrix, I observe that $CRP>0$ in 156 and $CRP<0$ in 144 of the 300 monthly estimates. Hence the requirement did not hold for 48\% of the time.} This rejection of mathematical feasibility is also the main critique point stated in \cite{Numpacharoen2013}, who provide a workaround for the $CRP<0$ cases. In this model, whenever $\alpha > 0$ the matrix $\mathbf{1}_{n\times n}$ of Eq.\ref{eq:CQ} is replaced by the 'equicorrelation lower-bound' $L$ defined as $\forall i\neq j: L_{ij} = \textrm{-}1/(n-1)$ and unit diagonal. Note that $L$ is simply the smallest possible \psd\ equicorrelation matrix (cp. the discussion on equicorrelations above). Following, in the \cite{Numpacharoen2013} modification, the implied correlation matrix is computed by
	\begin{align}
		\label{eq:expo}
		C_\mathbb{Q} = \begin{cases}
			\hat{\alpha} \mathbf{1}_{n\times n} + (1-\hat{\alpha}) C_\mathbb{P}, \quad \text{for} \quad CRP \geq 0 \implies \hat{\alpha} \geq 0\\
			\hat{\alpha} L + (1-\hat{\alpha}) C_\mathbb{P}, \qquad \text{for} \quad \hat{\alpha} < 0
		\end{cases}
	\end{align}
	. \cite{Buss2012} discuss that consistent investor preferences require that all $\mathbb{P}$-correlation pairs are scaled in the same direction, under general risk-aversion, this means up if $CRP > 0$ and down when $CRP < 0$. This consistency clearly holds in \cite{Buss2012} and \cite{Numpacharoen2013} for the $\hat{\alpha}\geq 0$ cases. Analyzing the workaround of \cite{Numpacharoen2013} in greater detail, one recognizes that off-diagonal entries inside $L$ are slightly negative and close to zero (e.g., \textrm{-}0.01 for $n\!=\!100$ and -0.002 for $n\!=\!500$). This causes that for $\hat{\alpha}<0$, almost every negative $\mathbb{P}$-correlation pair will be up-scaled, while positive ones will be down-scaled. Therefore, the \cite{Numpacharoen2013} workaround repairs mathematical feasibility of the basic adjusted ex-post model, but at the same time causes an inconsistent implementation of the correlation risk premium such that economical feasibility is rejected, leaving an insufficient solution of the implied correlation matrix behind.
	While mathematical/economical flaws still preserve, compared to the other concepts in place the adjusted ex-post method seems to be the most realistic one. The nearest implied correlation algorithm as introduced in Section 3 can be applied upon these models to repair feasibility.
	
	\paragraph{Skewness Approach} Worth to mention, \cite{Kempf2015,Chang2011} use a CAPM-like model, introduce economic assumptions to cancel out mathematical relationships, and estimate implied correlations between a stock and the market portfolio by combining option implied volatilities with risk-neutral skewness. Since CAPM is a factor model, the estimates correspond to $X$ of Eq.\ref{eq:C} and their approach can thus be used for a solution to the implied correlation matrix. However, this approach does not coerce with market conditions (Eq.\ref{eq:equilibrium}), nor does it stick to the boundaries $C_{ij}\in[\textrm{-}1,1]$, so mathematical and economical feasibilities are ignored at this approach. On the other hand, the \psd\ condition is easily met as $C$ is generated from the factor-model (Eq.\ref{eq:C}).\\	
	
	\section{Solutions from Factor Structures}
	
	The method uses implied volatilities as input parameters, which can be estimated in various ways (see \cite{Guo2016} for a detailed discussion). Dependent on what kind of implied volatilities to use (e.g., centered vs. directly parameterized \cite{Azzalini2008}), the methodology is not limited to Pearson type correlation matrices and can also be used within more sophisticated option pricing models, adjusting for non-normal distributions. An example for the case of a multivariate variance-gamma process can be found at the appendix.\\
	
	The for financial markets typical correlation structure is created out of the multi-factor copula model described by \cite{Glasserman2007}, which is used in a nearest correlation matrix context by \cite{Borsdorf2010}. The factor generating core of the model follows a simple but intuitive definition, that is
	\begin{align}
		{\xi} = {X}{\eta} + {F} {\epsilon}
	\end{align}
	where ${\xi}\in\mathbb{R}^n$ describes a random vector, ${X}\in\mathbb{R}^{n\times k}$ capturing factor exposures, ${F}\in\mathbb{R}^{n\times n}$ a diagonal matrix and ${\eta}\in\mathbb{R}^k$ corresponding to the factor's magnitude. All three vectors, ${\xi}$, ${\eta}$ and ${\epsilon}\in\mathbb{R}^n$ are defined to have zero mean and unit variance; ${\eta}$ and ${\epsilon}$ are orthogonal, 
	\begin{align}
		E[\xi] 	= E[\eta] = E[\epsilon]=0, \qquad var(\xi) = var(\eta) = var(\epsilon) = 1, \qquad cov(\eta, \epsilon) = 0
	\end{align}
	. From this follows that 
	\begin{align}
		\label{eq:coco}
		cov({\xi}) = E[{\xi}{\xi}'] = {X}{X}' + {F}^2
	\end{align}
	. Since ${\xi}$ has unit variance, $cov({\xi})$ turns out to be a correlation matrix with the boundaries of
	\begin{align}
		\label{eq:feas}
		\forall i = \{1,...,n\}: \quad	\sum_{d=1}^k X_{i,d}^2 + F_{ii}^2 = 1 \quad \implies \quad  \sum_{d=1}^k X_{i,d}^2 \leq 1 
	\end{align} 
	such that ${X}$ is necessarily limited to $[\textrm{-}1,1]$. From the above equation follows, that every $X_{i,d}^2$ corresponds to the goodness-of-fit of stock $i$ explained by the the risk factor $d$, and ${F}_{ii}^2$ to the unsystematic correlation which can not be explained by the given set of risk factors. Note, that ${XX}'$ is positive semi-definite by construction. The same is true for the diagonal matrix ${F}^2$, hence the sum of ${XX}'$ and $F^2$ will also be \textit{psd}. From the notation follows that $X_{i,d}$ itself can be interpreted as the correlation of a stock $i$ to a risk factor $d$ (e.g., $X_{i,m} = corr(i,m)$ in CAPM). $F^2$ itself is of less interest for further modeling as it can be easily computed once a feasible solution to $X$ is found. Therefore, the factor-structured correlation matrix of Eq.\ref{eq:coco} can be equivalently rewritten in a form where $F^2$ is suppressed and $X$ remains the only unknown. This reformulation will be used in the specific context as the factor-structured implied correlation matrix $C(X)$, which is now denoted by
	\begin{align}
		\label{eq:co}
		C(X) = J \circ XX' + I
	\end{align}
	.\footnote{Note that other literature (e.g., \cite{Borsdorf2010}) writes $C(X) = XX' - diag(XX') + I$ with $diag(XX')$ as the diagonal matrix of $XX'$. This alternative definition yields the same $C(X)$, however, I believe the version of Eq.\ref{eq:co} to be more efficient in terms of notation (looks cleaner to me), and in terms of computation as $J$ can be precomputed (thus not iterated) while $diag(XX')$ cannot.} Recap that $J=\mathbf{1}_{n\times n}-I$, hence $J \circ XX'$ has zero diagonal, and adding the identity matrix ensures that $C(X)$ strictly has a unit diagonal. Since $C(X)$'s off-diagonal elements are generated from $XX'$ (\textit{psd} by construction), and $diag(C(X))=\mathbf{1}_{n} \geq diag(XX')$, it follows that $C(X)$ is \textit{psd} for every feasible $X$.\\
	
	Generating $C(X)$ from Eq.\ref{eq:co} thus ensures a mathematically feasible correlation matrix. For economic feasibility, the equilibrium conditions of Eq.\ref{eq:equilibrium} need to be met. Therefore, the options traded on the market- and its sub-indices restrict possible values of $X$, I will call those restriction the "market-constraints". Given $n_c$ many market-constraints, the general solution to the factor-structured implied correlation matrix now evolves as
	\begin{align}
		\label{eq:general}
		C(X) =  J \circ XX' + I \qquad \text{subject to} \qquad \begin{cases}
			X \in \Omega := \left\{X\in \mathbb{R}^{n\times k}\!: \sum_{d=1}^{k} X_{i,d}^2\leq1, \forall i=\{1,..., n\} \right\}\\
			\sigma_j^2  = w_j'\sigma C(X) \sigma w_j, \quad \forall j = \{1,...,n_c\}
		\end{cases}
	\end{align}
	. From this definition follows, that $\Omega$ is a closed convex set of real numbers. Note that the number of unknown correlation pairs of a correlation matrix is $n(n-1)/2$ and the number of observable implied volatilities is $n + n_c$. Empirically, it is given that $n + n_c \ll n(n-1)/2$ almost surely, hence many different numerical solution to $X$ exist. To come up with a reasonable choice of $X$, this paper argues in favor of two approaches. First, a purely quantitative approach computing the feasible implied correlation matrix that is nearest to some pre-specified target. Second, an economic approach using factor-based asset pricing models like CAPM (\cite{Sharpe64}) or Fama-French (\cite{FF3}), assuming that $\mathbb{P}$ expected correlations to risk-factors can be estimated.\\
	
	The set $\Omega$ can be expressed as inequality constraints, and the market-constraints are of type equality. From the equality constraints follows that the feasible set is a subset (hypersurface if $n_c\!=\!1$) within $\Omega$, denoted by $\dot{\Omega}$ in the subsequent. For optimization purposes, one may consider the following definitions.
	\begin{definition}
		Let $h(X)$ define the function corresponding to the inequality constraints,
		\begin{align}
			h(X) := \mathbf{1}_{n \times 1} - (X \circ X)\mathbf{1}_{k \times 1}, \qquad h(X): \mathbb{R}^{n\times k}\rightarrow \mathbb{R}^{n}
		\end{align}
		such that $h(X)\geq \mathbf{0}_{n\times1}$ guarantees that $X\in\Omega$, which is necessary and sufficient for mathematical feasibility.
	\end{definition}
	
	\begin{definition}
		\label{def:g}
		Let $g(X)$ define the vector of market constraints, which simply stacks Eq.\ref{eq:equilibrium} for every (sub-)index that has option contracts traded upon,
		\begin{align}
			\forall j =\{2, ..., n_c\}: \quad g(X)  = \begin{pmatrix}
				\sigma_m^2\\
				\sigma_j^2\\
				\vdots
			\end{pmatrix} - \begin{pmatrix}
				w' \sigma C(X) \sigma w\\
				w_j' \sigma C(X) \sigma w_j\\
				\vdots
			\end{pmatrix}, \qquad g(X):= \mathbb{R}^{n\times k}\rightarrow \mathbb{R}^{n_c}
		\end{align}
		with $j=1$ reserved for the market portfolio $m$. Hence $g(X) = \mathbf{0}_{n_c\times 1}$ is necessary for economic feasibility. 
	\end{definition}
	The equality constraint is necessary, but not sufficient for economic feasibility. The best example is the (local-) equicorrelation model, which meets this requirement but yields in an unrealistic implied correlation matrix, such that economic feasibility is not met. Stacking constraints at Definition \ref{def:g} follows an economic intuition, but $C(X)$ is non-uniquely represented which potentially increases computational effort for large $n, n_c$. An alternative representations of $g(X)$ with $C(X)$ entering uniquely is
	\begin{align}
		g(X) &= \begin{pmatrix}
			\sigma_m^2\\
			\sigma_j^2\\
			\vdots
		\end{pmatrix} - [I \circ (W' C(X) W)] \bm{1}_{n_c} \qquad \text{where} \qquad W = \begin{pmatrix}
			w & w_j & \dots
		\end{pmatrix}, \quad \forall j = \{2, \dots, n_c\}
	\end{align}
	This representation computes a covariance matrix of constraints in a first step, where off-diagonal elements are eliminated in a second step by the Hadamard multiplication with $I$.\footnote{Alternatively to this, one could also vectorize $g(X)$'s right part of Definition 3.2 to  $g(X) = \begin{pmatrix}
			\sigma_j^2\\
			\vdots
		\end{pmatrix} - 
		\begin{pmatrix}
			w_j' \otimes w_j'\\
			\vdots
		\end{pmatrix} (\sigma \otimes \sigma)vec(C(X)), \forall j=\{1,...,n_c\}$. However, the Kronecker multiplication yields into allocation of very high-dimensional vectors, which potentially causes that storage limits are exceeded.}

	\subsection{Quantitative Approach: Computing Nearest Implied}
	
	Consider one has an educated guess of what the forward-looking correlation matrix potentially looks like. For example, this educated guess could be derived from historically realized correlations, from a GARCH forecast, or from a not fully feasible estimate of the discussed models in Section 2. As before, this educated guess is denoted by $A$, but different to \cite{Buss2012, Numpacharoen2013}, the assumption that $A=C_\mathbb{P}$ does not necessarily hold. With such an educated guess one can now search for a $C(X)$ that is as similar as possible to $A$,  but satisfies all market constraints in order to be considered a feasible implied correlation matrix. 
	
	\subsubsection{Formulating the Problem}
	
	Finding the most similar of a target matrix is known in mathematical literature as the nearest correlation matrix problem (e.g., \cite{Borsdorf2010}). Similarity between the generated $C(X)$ and the target matrix $A$ can be quantified by the squared Frobenius norm between them, 
	\begin{align}
		f(X) = \left\|C(X) - A \right\|_F^2, \qquad f(X): \mathbb{R}^{n \times k} \rightarrow \mathbb{R}_+
	\end{align}
	which, when introducing $\hat{A} = A -I$, can also be written as
	\begin{align}
		f(X) = \left\|J\circ XX' - \hat{A} \right\|_F^2
	\end{align}
	\begin{lemma} The gradient of $f(X)$ is
		\begin{align}
			\label{eq:grad}
			\nabla_X f = 4(J\circ XX' - \hat{A})X, \qquad \nabla_X f \in \mathbb{R}^{n\times k}
		\end{align}
	\end{lemma}
	\begin{proof}
		For simplicity, let $M:=J\circ X X' - \hat{A}$ and the Frobenius product (trace operator) be denoted by the colon $:$, $tr(M'M) = M\!:\!M$. The problem can then be written as $f(X) = \left\|M \right\|_F^2 = M\!:\!M$ and its differential is thus $df = 2M\!:\!dM$. The differential of $M$ itself is $dM = J\circ d(XX')$, and $d(XX')=dXX' + X dX'$. Both, $:$ and $\circ$ are mutually commutative operators, hence $M\circ J = J \circ M, \; M\!:\!J=J\!:\!M, \; M\!:\!J\circ d(XX') =  M\circ J \!:\! d(XX')$. Therefore, substituting back in yields
		\begin{align}
			df = 2 J \circ M:(dXX'+ X dX') = 2 (J\circ M + J'\circ M'):dXX',
		\end{align}
		since $J$ and $M$ are symmetric, $J = J'$ and $M = M'$, this reduces to $df = 4 J\circ M X : dX$, and
		\begin{align}
			\nabla_X f = \frac{\partial f}{\partial X} = 4 J \circ M X
		\end{align}
		follows. As $J\circ J = J$ and $J \circ \hat{A} = \hat{A}$, plugging back in $M$ gives Eq.\ref{eq:grad}.\footnote{This result is equivalent to the gradient as found in \cite{Borsdorf2010}, who derive it as $\nabla_X f = 4(XX'- diag(XX') - \hat{A} )X$}
	\end{proof}
	Popular optimization methods (like SQP) build on the Lagrangian function, which for the problem at hand can be formulated as 
	\begin{align}
		\mathcal{L}(X, \lambda, \kappa) = f(x) + \lambda'g(x) + \kappa'h(x))
	\end{align}
	with $\lambda\in \mathbb{R}^{n_c}$ and $\kappa \in \mathbb{R}^{n}$ representing the Lagrangian multipliers. Working with the Lagrangian is probably the most common practice, hence respective gradients are also reported below.
	
	%
	\begin{lemma} The gradient of $\lambda' g(X)$ with respect to $X$ is
		\begin{align}
			\label{eq:grad2}
			\nabla_X (\lambda' g(X)) = 2\sigma\left(J \circ \sum_{j=1}^{n_c} \lambda_j w_jw_j'\right)\sigma X
		\end{align}
	\end{lemma}
	\begin{proof}
		$\lambda'g(X)$ can be written as $\sum_{j=1}^{n_c}\lambda_j g_j(X)$. Focus on one market constraint $g_j(X)$, rearranged into the form $g_j(X) = \sigma_j^2 - w_j'\sigma (J \circ XX')\sigma w_j - w_j'\sigma^2 w_j$. To abbreviate notation, let $\hat{M}:=J\circ XX'$ and $B:=\sigma w_j w_j' \sigma$ as symmetric matrices (i.e., $\hat{M}'=\hat{M}$ and $B'=B$). For the gradient $\nabla_X g_j$, all non-$X$ terms cancel out so the focus is on $-w_j'\sigma \hat{M} \sigma w_j$, which can be expressed as the matrix trace $-tr(B \hat{M})$. The differential of $g_j$ is thus reduced to $dg_j= -tr(B d\hat{M}) = B : d\hat{M}$. From the proof of Lemma 3.1 it follows that $d\hat{M} = J\circ (dXX' + XdX')$. Therefore, the differential writes
		\begin{align}
			dg_j = -B \circ J :(dXX'+XdX') = -2 (B\circ J):dXX' = -2(B\circ J)X : dX
		\end{align} 
		hence the gradient of $g_j$ is given as $\nabla_X g_j = -2(B\circ J)X$. Note that $B\circ J$ equals $\sigma(J\circ w_j w_j')\sigma$, so substituting back in and multiplying by $\lambda_j$ gives
		\begin{align}
			\lambda_j\nabla_X g_j = -2 \lambda_j \sigma (J\circ w_j w_j')\sigma X
		\end{align}
		thus, the lemma evolves from $\nabla_X( \lambda'g(X) )= \sum_{j=1}^{n_c}\lambda_j\nabla_X g_j$.
	\end{proof}
	\begin{lemma}
		Let $D_\kappa$ denote the diagonal matrix of $\kappa$, the gradient of $\kappa'h(X)$ can then be written as
		\begin{align}
			\nabla_X(\kappa'h(X)) = -2 D_\kappa X
		\end{align}
	\end{lemma}
	\begin{proof}
		The term $\kappa'h(X)$ can be alternatively written as $\kappa'\mathbf{1}_{n\times 1} -\kappa'X^{\circ 2} \mathbf{1}_{k\times 1}$, where $X^{\circ 2}$ means that each element in $X$ is squared (Hadamard quadratic). Thus, $\partial (\kappa'h(X))/\partial x_{id} = -2\kappa_i x_{id}$, which  for $i\in\{1,...,n\}, d\in\{1,...,k\}$ yields into the result above.
	\end{proof}

	\subsubsection{Numerical Method}
	
	The nearest implied correlation matrix can now be attained by the following optimization:
	\begin{align}
		C(X) = J \circ XX' + I \quad & \text{with} \quad X = \argmin_{X} \mathcal{L}(X,\lambda,\kappa)
	\end{align}
	As the objective function $f(X)$ is not convex, solutions found are expected to be local minima. A performance comparison of numerical methods for a general nearest correlation matrix problem (i.e., without equality constraints $g(X)$) was made by \cite{Borsdorf2010}. Their recommendation is, that the spectral projected gradient (SPGM) is the most efficient and reliable method for this task.\footnote{Compared algorithms include 'alternating directions', 'principal factors method', spectral projected gradient method' and Newton-based methods, see \cite{Borsdorf2010} for details.} Based on this finding, the SPGM is also used as main algorithm within this study for solving the implied correlation matrix. A brief comparison to a SQP based solver is included at the empirical section.
	
	SPGM was initially introduced by \cite{Barzilai1988} and is used in many different studies, \cite{Birgin2014} summarize a long list of corresponding works, applications to financial data include \cite{Higham2002, Borsdorf2010}. A detailed discussion on the algorithm can be found in \cite{Birgin2001,Birgin2014}. Technically, the SPGM comes with three advantages when compared to other potential candidates. First, the method 'only' requires the gradient - known for this task -, but not the Hessian matrix. Second, the non-linear constraints are  guaranteed to be satisfied and third, the method is ensured to converge towards the optimum (cp. \cite{Barzilai1988}). In a nutshell, the method minimizes a continuous differentiable function on a nonempty closed convex set via projecting possible values of $X$ back onto the feasible set. Consequently, a key requirement is that projections on the feasible set can be made at low computational effort. If so, then the method provides a very efficient way of handling the constraints. For the application at hand, handling the equality constraint alongside $\Omega$ can become tricky. This issue can be resolved by combining SPGM with a more general inexact restoration framework (e.g., \cite{Martinez2000} for a discussion). The IR-SPGM was introduced by \cite{Gomes2009} who also provide an in-depth explanation of the algorithm, hence the details will not be repeated here. Roughly speaking, the algorithm can be broken down into two phases: a restoration phase, projecting $X$ back onto the feasible set, and an optimization phase, computing the step size. For the particular application at hand, greater attention has to be paid for the projection function. Recap that $\dot{\Omega}$ defines the feasible region that satisfies both $\Omega$ and the market-constraints. As mentioned by \cite{Gomes2009}, for some $X$ outside $\dot{\Omega}$, the projection function $P_{\dot{\Omega}}(X)$ should be defined such that $P_{\dot{\Omega}}(X) \approx \argmin_{\dot{X}} \| X - \dot{X}\|_2$ where $\dot{X}\in {\dot{\Omega}}$. This means that the projection function should find a point in $\dot{\Omega}$ that is approximately the closest to $X$ (inexact restoration). Obviously, $P_{\dot{\Omega}}(X)$ is an orthogonal projection if an $\dot{X}\perp X$ exists, which is not necessarily the case given the upper boundaries of $\dot{X}$ (Eq.\ref{eq:feas}). Fig. \ref{fig:omega} visualizes the projection and feasible set on a simplified example.
	\begin{figure}[h]
		\centering
		\includegraphics[scale=0.7]{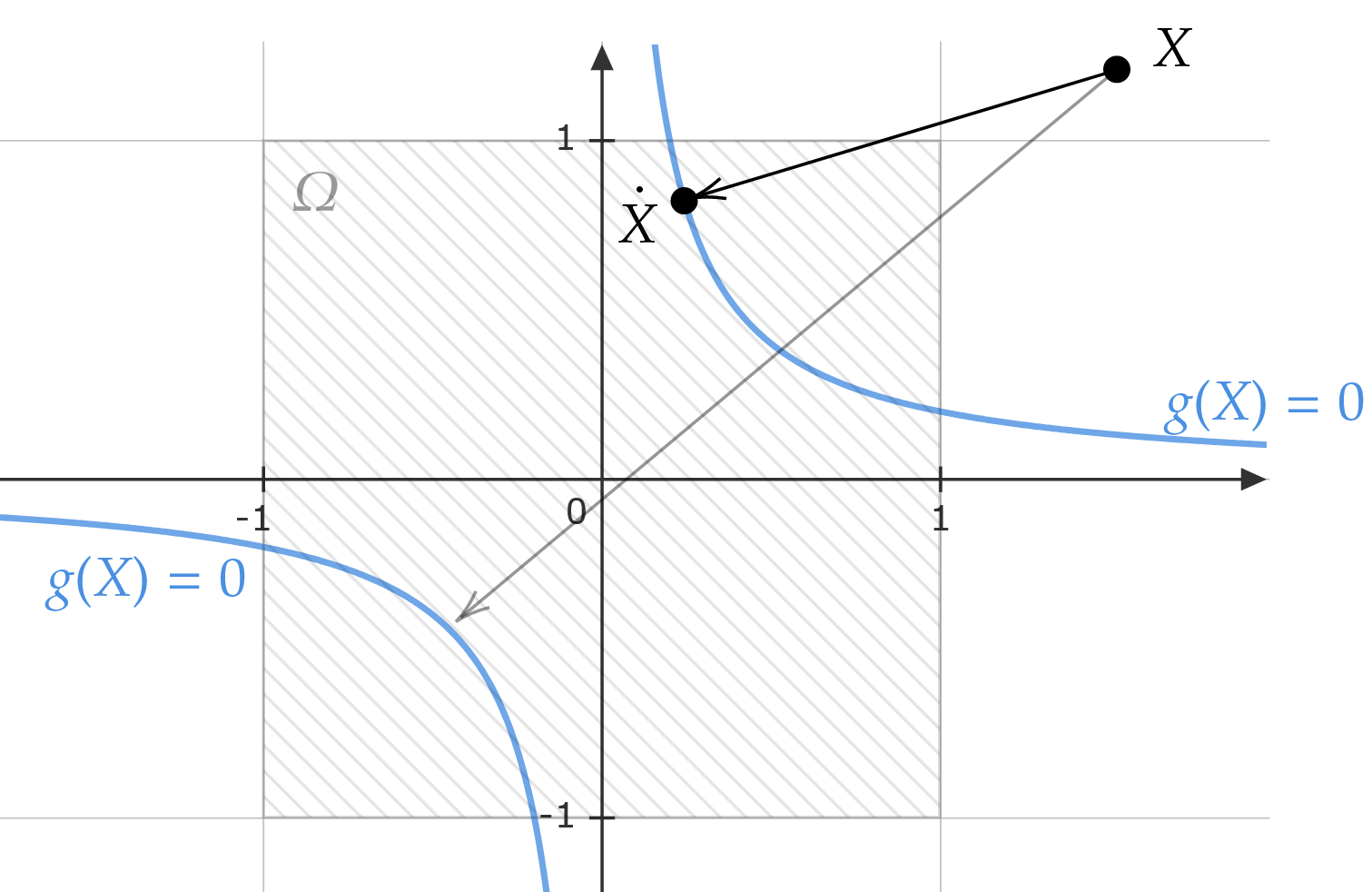}
		\caption{Visualization of the constraints for the two-assets/one-factor/one-market-constraint case. The inequality constraint $\Omega$ spans the gray box of technically feasible solutions, the blue line defines the solutions which satisfy the market constraint. As can be seen, the market constraint actually consists of two convex curves. Hence, two orthogonal projections of $X$ onto $g(X)=0$ exist, but only one (i.e., $\dot{X}$) has minimum distance to $X$.}
		\label{fig:omega}
	\end{figure}
	
	For the inexact restoration framework the projection is split up into two functions, a projection onto the inequality constraint $P_\Omega(\cdot)$ and one onto the equality constraint $P_E(\cdot)$. As for $P_\Omega(\cdot)$, \cite{Borsdorf2010} already discuss that this projection can be easily carried out by replacing every row $i$ of $X$ which exceeds $\sum_{d=1}^{k} X_{i,d}^2 > 1$ by $X_i/\|X_i\|$. So the projection on $\Omega$ comes at very low computational effort. 
	Let $\sigma_m^2$ induce the only market-constraint and $X_E$ is any point that matches it such that $g(X_E)=0$. Given so, one can formulate the Lagrangian as
	\begin{align}
		\mathcal{L}_{E}(X_E, \lambda_E)=\| X_E - X \|_F^2 - \lambda_E [v' (X_EX_E'\circ J + I)v - \sigma_m^2]
	\end{align}
	. To simplify notation, introduce $v=\sigma w$. With the results from above, the gradient of $\mathcal{L}_E$ can now be written as
	\begin{align}
		\nabla_{X_E} \mathcal{L}_{E} = 2(X_E - X) - 2{\lambda}_E(vv'\circ J)X_E \qquad \text{and} \qquad \nabla_{{\lambda}_E} \mathcal{L}_{E}= \sigma_m^2 - v' (X_EX_E'\circ J + I)v 
	\end{align}
	. To minimize the Lagrangian function, both gradients are set to zero, $\nabla_{X_E} \mathcal{L}_{E} = \nabla_{\lambda_E} \mathcal{L}_E=0$. Therefore, from rearranging terms of $\nabla_{X_E} \mathcal{L}_{E}$ one can write
	\begin{align}
		\quad {X}_E = (I-{\lambda}_E vv' \circ J)^{-1}X
	\end{align}
	The inverse of $(I-\lambda_E vv' \circ J)^{-1}$ can be expressed following the expansion procedure of \cite{Miller1981}. Since in a real economy the elements of $vv' \circ J$ are generally smaller 1 and close to 0,\footnote{If this is not the case, implied volatilities can be time-scaled down e.g. from yearly to daily such that $\forall i: v_i < 1$.} higher order terms like $(vv'\circ J)^2$ converge towards zero very fast. Consequently, the inverse is efficiently approximated by
	\begin{align}
		(I-{\lambda}_E vv' \circ J)^{-1} \approx I+{\lambda}_E vv' \circ J\\
		\label{eq:xe}
		\implies P_E(X): X_E\approx X + {\lambda}_E (vv' \circ J) X
	\end{align}
	. Plugging back into the equality constraint $\nabla_{\lambda_E} \mathcal{L}_E=0$ and rearranging terms, ${\lambda}_E$ can then be solved from a simple quadratic equation having two solutions,
	\begin{align}
		{\lambda}_{E,\pm}= \frac{{\lambda}_{E,1} \pm \sqrt{{\lambda}_{E,1}^2 - 4 {\lambda}_{E,0}{\lambda}_{E,2}}}{2{\lambda}_{E,0}}, \qquad \text{with} \qquad \begin{cases}
			{\lambda}_{E,0}  = v'[(vv'\circ J)^2\circ J]v \\
			{\lambda}_{E,1} = 2v'[XX' (vv'\circ J)\circ J]v\\
			{\lambda}_{E,2}= v'[XX'\circ J + I]v - \sigma_m^2
		\end{cases}
	\end{align}
	. Having estimated both solutions of ${\lambda}_{E,\pm}$ and inserting the results into Eq.\ref{eq:xe}, it is easy to attain whether ${X}_{E,+}$ or ${X}_{E,\textrm{-}}$ is closer to $X$. Hence, it evolves that also the projection $P_{E}(X)$ is inexpensive to compute. In case ${\lambda}_E$ is complex, taking only its realistic part is sufficient given the inexact restoration framework.\\
	$\dot{X}$ that fulfills both $\dot{X}\in \Omega$ and $g(\dot{X})=0$ can now be found following a simple alternating projection algorithm. Hereby, start with the initial projection on the market-constraint $P_E(X)$. If $P_E(X)\in\Omega$ then $\dot{X}$ has already been found. Otherwise, lock-in whether $\lambda_{E,+}$ or $\lambda_{E,\textrm{-}}$ was used and take only the upper or lower solution subsequently. This is useful because thereby the market-constraint becomes convex and alternating now the projections $P_\Omega(\cdot)$ and $P_E(\cdot)$ converges towards the optimal solution $\dot{X}$. The alternating procedure can be stopped if a certain tolerance level for $g(\cdot)$ has been reached.
	
	With the defined projection procedure, the optimization under IR-SPGM can now be carried out following the guidelines of \cite{Gomes2009}.

	\subsection{Economic Approach}
	Alternatively to the above quantitative approach, one can also pre-specify expected factors $X$. The factor exposures can be computed either from a statistical (e.g., prinicpal component analysis) or from an economic routine (like CAPM, Fama-French \cite{FF3}). As for the latter, $X$ corresponds to the correlation between stocks and the risk factors. 
	Note, from literature on factor analysis it is known that if a $X$ can be found such that $XX'$ has unit diagonal, then a correlation matrix is fully explained by its reduced structure $X$. It is also known that if $k\!=\!n$, then a correlation matrix can be fully described by a set of factors (something we know e.g. from eigen-value decomposition). So, either way, a (\psd\ ) correlation matrix $C$ is an aggregate with $X$ as the underlying structure. In the models of \cite{Buss2012,Numpacharoen2013}, the correlation risk premium was incorporated on the aggregated level $C$; different to that, in the subsequent the premium will be incorporated on the underlying structure $X$.\\
	
	Consider the implied correlation matrix is described by a one-factor model, $k\!=\!1$.\footnote{Actually, in financial markets it is documented (e.g., \cite{...}) that the first principal component already explains a very large portion of the realized correlation matrix. It is also reported that the first principal component is very similar to CAPM's market portfolio, see e.g. \cite{...}.} Given so, the weighting idea of \cite{Buss2012} can be applied on the factor loading level to achieve the $\mathbb{P}-\mathbb{Q}$ transformation. As in \cite{Buss2012, Numpacharoen2013}, assume that physically expected correlations can be estimated. Next, following their regime, $X_\mathbb{P}$ will weighted against $\mathbf{1}_{n\times k}$ if $CRP>0$. In case $CRP<0$, then the weighting is made towards the lower boundary, which can now be chosen as $\textrm{-}\mathbf{1}_{n\times k}$. Hence, there is not an economically inconsistent-scaling as in \cite{Numpacharoen2013}, and there is also no \psd\ problem as in \cite{Buss2012}. Introduce the sign of the correlation risk premium as $\upsilon:=\text{sign}({CRP})$ to indicate up- or down-scaling, the transformation can then similarly to Eq.\ref{eq:CQ} be written as
	\begin{align}
		\label{eq:xq}
		X_\mathbb{Q} = X_\mathbb{P} + \tilde{\alpha} X_{\Delta},\qquad \text{with}\qquad  X_{\Delta} = \upsilon\mathbf{1}_{n\times k} - X_\mathbb{P}
	\end{align}
	. It generally holds that $\tilde{\alpha}\in[0,1]$ and the implied correlation matrix is given by Eq.\ref{eq:C} as $C=X_\mathbb{Q}X_\mathbb{Q}' \circ J + I$. Note that $C$ can now cover the entire range of mathematically feasible correlation matrices, as it is generated from Eq.\ref{eq:C}. To compute $\tilde{\alpha}$, one has to match the market constraint
	\begin{align}
		\sigma_m^2 = w' \sigma (X_\mathbb{Q}X_\mathbb{Q}'\circ J + I) \sigma w
	\end{align} 
	, from which it follows that $\tilde{\alpha}$ can be explicitly solved for by
	\begin{align}
		\label{eq:alpha}
		\tilde{\alpha} = \frac{\textrm{-}\sigma_{\mathbb{P},\Delta}^2 + \upsilon\sqrt{\sigma_{\mathbb{P},\Delta}^4 - \sigma_{\Delta}^2(\sigma_{\mathbb{P}}^2 - \sigma_m^2 )} }{\sigma_{\Delta}^2} \qquad \text{with} \qquad \begin{cases}
			\sigma_{\mathbb{P}}^2 &= w'\sigma(X_\mathbb{P}X_\mathbb{P}' \circ J + I) \sigma w\\
			\sigma_\Delta^2 &= w'\sigma(X_{\Delta}X_{\Delta}' \circ J ) \sigma w\\
			\sigma_{\mathbb{P}, \Delta}^2 &= w'\sigma(X_\mathbb{P}X_{\Delta}' \circ J) \sigma w
		\end{cases}
	\end{align}
	.\footnote{For multiple market constraints, one can either calibrate $X_\mathbb{P}$ and obtain a unique $\tilde{\alpha}$, or first solve for $\tilde{\alpha}$'s on sub-index levels, and then on the market index level.} The above quadratic equation obviously gives two values of $\tilde{\alpha}$, given that "$+\upsilon$" can be replaced by "$\pm$". Since $\tilde{\alpha}$ is associated to an economic interpretation, the upper value is taken in case of $CRP > 0$ (up-scaling), and the lower value if $CRP < 0$ (down-scaling), hence "$\pm$" is substituted by "$+\upsilon$" to implement this mechanism. 
	In a multi-factor pricing model (like Fama-French three factor \cite{FF3}), high multicolinearity among factors could potentially violate $X_\mathbb{P}\in \Omega$. This issue can be easily resolved by first computing the correlations between stocks to factors, and then orthogonalize the set of vectors (e.g., via Gram-Schmidt process).\\
	
	The procedure of computing the implied correlation matrix from pricing factors can now be summarized as follows. First, compute physically expected correlations between stocks to risk factors, denote this set of vectors as $X_\mathbb{P}$. If necessary, orthogonalize $X_\mathbb{P}$. Second, calculate the weighting scalar $\tilde{\alpha}$ according to Eq.\ref{eq:alpha}. Third, compute $X_\mathbb{Q}$ from Eq.\ref{eq:xq} and the implied correlation matrix then evolves as $C=X_\mathbb{Q}X_\mathbb{Q}' \circ J + I$.

	\section{Empirical Experiment}
	A brief empirical experiment on computation of the factor structured implied correlation matrix is drawn from data of the S\&P 100 and S\&P 500 index to evaluate implementation difficulty and computational efficiency. 
	The focus is on monthly at-the-money (ATM) Call option implied volatilities with a target maturity of one month, directly derived from OptionMetrics. Note that OptionMetrics follows a three-dimensional kernel regression for interpolating the option surface (see \cite{OptionMetrics}), a discussion on this topic can be found in \cite{Ulrich2020}. Computations are carried out at the beginning of each month, starting in 1996-01-01 to 2020-12-02 (i.e., the maximum of available data for this study), thus giving 300 implied correlation matrix estimates per time-series. The ATM level is chosen due to three reasons: first, options are typically most liquid around the ATM level (\cite{Etling2000}); second, the ATM level is less sensitive to model misspecification (compared to out-of-money; \cite{Carr2003}); third, ATM Call prices are close to ATM Puts. Return (daily) and market value (monthly) data are derived from CRSP (Center for Research in Security Prices), the lists of index constituents are from Compustat. All computations were executed under a Intel\textsuperscript{\textregistered} Core\texttrademark\ i5-8250U CPU with 1.60GHz, using the statistical programming software \textsf{R}. As for the optimization algorithms, a variance tolerance for the market constraint of 1e-06 is chosen. The stopping criteria is set to a marginal improvement in the objective function of 1e-03. The IR-SPGM is implemented as described by \cite{Gomes2009} (Algorithm 2.1), except that a monotone line search strategy is used.\footnote{The non-monotone line search strategy runs additional sub-routines of projections to speed-up convergence. But since the projection function here is a potential alternating procedure on its own, I found the monotone line search to be slightly faster than the non-monotone one.}
	Two types of target matrices are used. First, simple Pearson's correlation matrices from 12-month historically realized returns. Second, a mean-reverting matrix with the entries $\hat{\rho}(i,j) = \theta_{ij}\rho(i,j) + (1-\theta_{ij})\bar{\rho}(i,j)$, using 9 months of historically realized returns for $\rho(i,j)$ and the mean-correlation between $i$ and $j$ over the total time-horizon for $\bar{\rho}(i,j)$. The reversion speed $\theta_{ij}$ is randomly drawn from a uniform-distribution between 0 and 0.4 to bring in some noise into the target matrix. Generally, the two S\&P indices are rebalanced on a quarterly basis, hence there is no guarantee that the target matrices per se are \psd\ . In such cases, the factor method also served as a repairing tool. 
	
	As for the starting value $X^{(0)}$, I follow a modified approach as proposed by \cite{Borsdorf2010}. For the target matrix $A$, let $e$ be the set of eigenvectors and $\iota$ the corresponding eigenvalues. Then, for each column $d=\{1,...,k\}$ of $X^{(0)}$, compute the starting values as
	\begin{align}
		\forall d = \{1,...,k\}: X_d^{(0)} = \varsigma_d e_d \quad \text{where} \quad \varsigma_d=\min\left\{\sqrt{\frac{(\iota_d - 1)\|e_d\|_2^2}{k\|e_d\|_2^4 - k \sum_{i=1}^n e_{d,i}^4 }} ,  \frac{1}{\sqrt{k} \max_i |e_{d,i}|}\right\}
	\end{align}
	. This starting value is identical with the one proposed by \cite{Borsdorf2010} in case $k=1$, but differs when the number of factors is higher. Within the empirical experiments of the $k>1$ cases I observe that the reduction in the objective function is larger under the modified starting value, hence the proposed version of $X^{(0)}$ is preferred.\\
	
	
	The empirical experiment is now split into three parts, which are summarized in Tab.\ref{tab:res} Panel A, B and C. The columns of Tab.\ref{tab:res} describe the number of risk factors (\textsf{k}), the target matrix (\textsf{A}), the mean and standard deviation of computation time (\textsf{t}), the optimized objective function (\textsf{fn}), the absolute realization of the variance tolerance (\textsf{|v.tol|}) and the number of outer iterations (\textsf{iter}). The first two panels show results on the quantitative approach, computing the nearest implied correlation matrices. The third panel shows results on the economic approach.
	
	In the first part (Panel A), the nearest implied correlation matrix is computed using IR-SPGM for the S\&P 100 under different settings with respect to the number of risk factors and the two different target matrices. Overall, the patterns between the historical- and the mean-reverting matrices look very similar. I observe that the average computation speed is very fast and comparable to the un-market-constrained results of \cite{Borsdorf2010}. This observation is probably due to the fast convergence within few iterations - around 3 to 5 on average - and the simplicity of the projection. As it is unlikely that the target matrix perfectly equals the hidden true implied correlation matrix, it is not expected that the objective function will converge towards zero as the mark-up due to a correlation risk premium remains. At this point it is found that an increasing number of risk factors indeed reduces the final objective function, hence improves the estimation accuracy. On the other hand, non-surprisingly it is also observed that a larger $k$ comes at higher computational effort as it multiplies the model's number of variables. With a look on the variance tolerance, the IR-SPGM algorithm had no difficulties to stay inside the feasible region.
	
	The second part (Panel B) applies the nearest implied correlation matrix method on three alternative settings. On the first line, computation results on S\&P 500 data are shown, increasing the number of unknown correlation pairs from 4950 ($n\!=\!100$) to 124'750. With the increase of number of stocks the computation time increased over-proportional, which is probably due to the larger average number of iterations needed to converge towards the optimum (3.037 vs. 5.003). The optimized objective function can be compared when dividing by $2n_\rho$, which is 0.016 for the S\&P 100 and 0.011 for the S\&P 500 index. With a look on the variance tolerance, the S\&P 500 computations did stick more strictly to the market constraint with a maximum deviation of 1.3e-10. Therefore, while the computation time over-proportionally increased from the S\&P 100 to the 500 index, the mean \textsf{fn} per matrix-entry and also the realized variance tolerance were remarkably smaller. Hence, stopping criteria and variance tolerances can be potentially relaxed the greater the index. The second line of Panel B reports results on the S\&P 100 where the target matrix is chosen from mean-reverting correlations that are converted into an implied correlation matrix according to the model of \cite{Buss2012} (Eq.\ref{eq:CQ}). Hence, the target matrix here already is an implied correlation matrix that matches the market constraint, but does not stick to mathematical feasibility. At this setting, only the monthly estimates where the target matrix was not \psd\ were taken, such that the factor method is used to repair the estimates under the \cite{Buss2012}-model. This included 64 of the 300 monthly matrices. The number of risk factors is set higher here to achieve a better fit, and with a look on \textsf{fn} one observes that the objective function is indeed substantially smaller at this application. Hence the factor model qualifies as a repair tool for existing implied correlation models. The third line of Panel B replaces the IR-SPGM algorithm by a sequential quadratic programming method (SQP) for the S\&P 100 data. The algorithm used is borrowed from the \textsf{Rsolnp} package (see \cite{Ghalanos2015,Ye1989}). Using the SQP solver serves as a reference to cross-validate whether the IR-SPGM method was correctly implemented, the results are thus directly compared to the first line of Panel A. What I observe is that the optimized objective functions are very similar between the SQP and IR-SPGM algorithm, hence I conclude that the IR-SPGM algorithm was correctly implemented. Comparing computation times between them, SQP seems to be not competitive, which was also found on the un-market-constrained case of \cite{Borsdorf2010}. This finding thus motivates the usage of IR-SPGM.
	
	Panel C of Tab.\ref{tab:res} implements the economic approach on the S\&P 500 index with respect to CAPM (market factor), the Fama-French three factor model (\cite{FF3}), the extension for the momentum factor (\textsf{FF3+Mom.}) and the Fama-French five factor model (\cite{FF5}). Return data on the factor portfolios is derived directly from K.R. French's data library.\footnote{\url{http://mba.tuck.dartmouth.edu/pages/faculty/ken.french/data_library.html}} To approximate $X_\mathbb{P}$, 12-month realized correlations between stocks and the risk factors are used. Next, to ensure that $X_\mathbb{P}\in\Omega$, the factors are orthogonalized via Gram-Schmidt process.\footnote{In the following order: Market, SMB, HML, RMW, CMA, Momentum.} Since the economic approach does not iterate and consits only of one projection, computation turns out to be very inexpensive. On the other hand, one also recognizes that it deviates more strongly from the target matrix than the quantitative approach does, but still the larger the number of risk factors the smaller \textsf{fn.mean}. Similarity to the target matrix, however, plays an subordinate role here. A more interesting pattern evolves with an eye on $\hat{\alpha}$. Recap that $\hat{\alpha}$ represents the weight of the boundary and $(1-\hat{\alpha})$ the weight of $X_\mathbb{P}$ inside the risk-neutral factor correlations $X_\mathbb{Q}$. When comparing $\hat{\alpha}$ now among the different pricing models, one observes that the average value declines. This means that the larger the number of risk factors, the less is the modification to match the observed implied market variance, hence a larger $k$ also explains more of the hidden (true) implied correlation matrix. Generally, as $\hat{\alpha}$\textsf{.mean} is close to zero across all four economic models, I conclude that correlation risk premia enter modestly in the factor-structured implied correlation matrix framework.
	
	\begin{table}[h]
		\centering
		\caption{Summary statistics of computing nearest (quantitative approach, Panel A-B) and pricing-factor structured (economic approach, Panel C) implied correlation matrices on monthly option data. As for the quantitative approach, it turns out that computations are carried out in a small amount of time, converging towards the optimal solution within few iterations. The nearest factor-structured matrix can thus be used either as a stand-alone estimate, or as a tool to repair a non-positive-semi-definite implied correlation matrix. Introducing economic assumptions, also pricing factors from models like CAPM or Fama-French can be used to estimate the implied correlation matrix.}
		\footnotesize
		\begin{tabular}{lccccccccccc}
			\toprule
			& k     & A     & t.mean & t.sd  & fn.mean & fn.sd & |v.tol|.mean & |v.tol|.max & iter.mean & iter.sd & index\\
			\midrule
			\textbf{Panel A:} &       &       & \multicolumn{9}{c}{\textit{hist. matrix}} \\
			SP100 & 1     & hist. & 0.051 & 0.024 & 159.9 & 172.8 & 5.1E-11 & 1.5E-08 & 3.037 & 1.491 & SP100 \\
			SP100 & 3     & hist. & 0.106 & 0.050 & 111.0 & 168.3 & 1.3E-08 & 6.2E-07 & 4.819 & 2.137 & SP100\\
			SP100 & 5     & hist. & 0.140 & 0.077 & 102.7 & 166.9 & 1.9E-08 & 9.9E-07 & 4.990 & 1.977 & SP100\\
			&       &       & \multicolumn{9}{c}{\textit{mean-reverting matrix}}\\
			SP100 & 1     & m.r.  & 0.055 & 0.027 & 158.0 & 158.2 & 2.0E-11 & 2.5E-09 & 3.064 & 1.438 & SP100\\
			SP100 & 3     & m.r.  & 0.107 & 0.047 & 113.3 & 173.2 & 3.9E-08 & 9.6E-07 & 4.708 & 1.790 & SP100 \\
			SP100 & 5     & m.r.  & 0.149 & 0.082 & 104.6 & 172.5 & 2.4E-08 & 8.9E-07 & 5.054 & 1.989 & SP100\\
			\midrule
			\textbf{Panel B: } &       &       &       &       &       &       &       &       &       & & \\
			SP500 & 1     & hist. & 5.347 & 2.311 & 2847.9 & 4395.0 & 1.5E-11 & 1.3E-10 & 5.003 & 1.892 & SP500\\
			Repaired $C_\mathbb{Q}$ & 15    & $\text{m.r.}_\mathbb{Q}$  & 0.223 & 0.070 & 16.41 & 6.775 & 3.1E-08 & 9.7E-07 & 7.184 & 2.351 & SP100\\
			SQP   & 1     & hist. & 0.337 & 0.098 & 161.4& 166.8 & 3.5E-07 & 9.6E-07 & 3.517 & 0.721 & SP100\\
			\midrule
			\textbf{Panel C:} &       &       &       &       &       &       &       &       & $\hat{\alpha}$.mean & $\hat{\alpha}$.sd & \\
			\cmidrule{10-11}    
			CAPM  & 1     & hist. & 0.017 & 0.018 & 5108.5 & 5858.9 & 3.0E-17 & 4.7E-16 & 0.138 & 0.133 & SP500\\
			Fama-Fr.3 & 3     & hist. & 0.019 & 0.018 & 4376.6 & 4740.8 & 3.3E-17 & 2.5E-16 & 0.091 & 0.070& SP500 \\
			FF3+Mom. & 4     & hist. & 0.019 & 0.019 & 4167.1 & 4381.3 & 3.3E-17 & 3.1E-16 & 0.082 & 0.060& SP500 \\
			Fama-Fr.5 & 5     & hist. & 0.019 & 0.024 & 4125.3 & 4509.1 & 3.6E-17 & 3.9E-16 & 0.076 & 0.055& SP500 \\
			
			\bottomrule
		\end{tabular}%
		\label{tab:res}%
	\end{table}%

	\section{Concluding Remarks}
	Having an idea about future diversification possibilities requires knowledge about future correlations. Identifying such is a challenging task as backward-looking time series will never capture information on events that will happen in the future. On the other hand, option implied volatilities are known to carry information on the expected development of a market, hence are used by academics and practitioners to bring in forward-looking perspectives. Computing a feasible implied correlation matrix, however, remains puzzling as it is highly under-determined, allowing for many possible solutions. This paper discusses mathematical and economical necessary conditions for feasible solutions to this problem and finds, that existing models are typically unable to guarantee so. As correlation matrices can be represented in terms of factors in general, this work rolls-up a solution from the underlying factor structure, which turns out to be an easy way of handling the mathematical requirements next to matching observable market-constraints. I conclude that the quantitative approach of computing the nearest factor-structured implied correlation matrix is a useful tool for repairing non-positive-semi-definite (implied) correlation matrices that can be also used as a stand-alone estimate, coming at a minimum of assumptions. With the economic approach it is demonstrated how expected correlations - or, potentially factor-betas - between stocks and the market portfolio can be used to translate into an estimate of the implied correlation matrix. Both approaches are empirically evaluated on monthly S\&P 100 and S\&P 500 option data (1996-2020), using a inexact-restoration spectral projected gradient method for the quantitative, and Fama-French factors for the economic approach. Since implementation and computation of the factor-framework turns out to be straight-forward, potential applications of the method are in the context of basket option pricing, estimating forward-looking betas with the corresponding implications of risk-/portfolio-management and equity valuation, or forecasting stock-market movements.

	
	
	\newpage
	\begingroup
	\raggedright
	\bibliographystyle{ieeetr}

	\addcontentsline{toc}{section}{References}
	\endgroup

	\section*{Appendix}
	\subsection*{Non-Gaussian Copula: Example of Variance-Gamma}
	Risk-neutral densities are typically of asymmetric shape and heavy tails, causing stocks to correlate higher for market down-turns and less for up-movements. The Pearson's correlation matrix cannot capture such asymmetries. While for at-the-money options there is almost no difference in implied correlations between the Gaussian or a more sophisticated copula (see e.g., \cite{Linders2016})\footnote{ATM option prices are typically very similar among most option pricing models, hence the same holds for implied volatilities and thus also for implied correlations}, it can indeed matter for out-of-money options. An easy extension for non-normal shapes/tails can be made if the transformation between direct and centered multivariate parametrization is known (following e.g., \cite{Azzalini2008}). To provide an example, the case of the variance gamma model (\cite{Madan1990, Madan1998}) is discussed below, which became quite popular for pricing basket options (\cite{Linders2016}).\\
	
	Let $Z(t)$ follow a one-factor multivariate variance-gamma process, constructed from subordinating the multivariate Brownian motion $B$ with the Gamma distributed subordinator $V(t)\sim Ga(t/\nu, 1/\nu)$,
	\begin{align}
		Z(t) = t {\mu} + {\theta}V(t) + B(V(t)) \qquad s.t. \qquad Z \sim \mathcal{VG}(\xi, \omega, C_{dir}, \theta, \nu)
	\end{align}
	with ${\xi}\in\mathbb{R}^n$ as the location, $\omega\in\mathbb{R}_+^n$ the diagonal matrix of scale , ${\theta}\in\mathbb{R}^n$ the shape, $\nu\in\mathbb{R}_+$ the variance rate and $C_{dir}$ the correlation matrix; all being direct parameters. The subscripts $dir$ and $cen$ are used for direct/centered matrix parametrization. The direct covariance-matrix is thus given by $\Sigma_{dir}=\omega C_{dir} \omega$. All direct parameters except $C_{dir}$ can be derived for the cross-section and the (sub-)index from option data, for example from FFT \cite{Madan1998}. Further, we know that the first two centered moments evolve as
	\begin{align}
		E[Z] = {\xi + {\theta}} \qquad \text{and} \qquad E[(Z - E[Z])^2] = {\Sigma}_{dir} + \nu {\theta}{\theta}' 
	\end{align}
	(cp. \cite{Madan1990}). Hence, within the VG model, the change between direct and centered parametrization is easily obtained by $\Sigma_{cen} = \Sigma_{dir} + \nu {\theta}{\theta}' \equiv \sigma C_{cen} \sigma$, with $\sigma$ still being the diagonal matrix of centered volatilities. The centered Pearson's correlation matrix $C_{cen}$ is thus given by 
	\begin{align}
		\label{eq:cen}
		C_{cen} = \sigma^{-1}\omega C_{dir} \omega \sigma^{-1} + \nu \sigma^{-1} \theta\theta'\sigma^{-1}
	\end{align}
	
	As before, it still holds that the portfolio variance is computed from $\sigma_m^2 = w'\Sigma_{cen}w$, which defines the market constraint. So substituting Eq.\ref{eq:cen} back into the market constraint and introducing $C_{dir}(X) = J \circ XX' + I$ allows to estimate the direct correlation matrix (i.e., non-Pearson) via the nearest implied correlation matrix (NICM) method. This fact demonstrates that the NICM method is not limited to Pearson's correlation matrices. Worth to mention, $C_{dir}(X)$ is \textit{psd} by construction, the same holds for $C_{cen}$ as $\nu\theta\theta'$ is \textit{psd} and $\sigma$ and $\omega$ are diagonal of positive entries. Hence $C_{dir}(X) = J \circ XX' + I$ ensures mathematical feasibility.
\end{document}